\documentclass[11pt,a4paper]{article}
\usepackage{amsmath,amsxtra,amscd,amssymb,latexsym,stmaryrd,amsthm}
\usepackage{graphicx,color}
\usepackage{hyperref}
\usepackage{appendix}
\usepackage[utf8]{inputenc}
\usepackage[T1]{fontenc}
\usepackage{mathrsfs}
\usepackage[all]{xy}
\usepackage{tikz,xcolor}
\usepackage{authblk}

\definecolor{Red}{cmyk}{0,1,1,0}

\definecolor{Blue}{cmyk}{1,1,0,0}

\theoremstyle{plain}
\newtheorem{theorem}{Theorem}[section]
\newtheorem{corollary}[theorem]{Corollary}
\newtheorem{proposition}[theorem]{Proposition}
\newtheorem{lemma}[theorem]{Lemma}

\newcommand{\ba}{{\bf a}}
\theoremstyle{definition}

\title{Pressure Derivative on Uncountable Alphabet Setting: a Ruelle Operator Approach}

\author[$\S$]{E. A. Silva}
\affil[$\S$]{ Departamento de Matem\'atica - UnB - 70910-900, Bras\'ilia, Brazil}
\date{\small\today}

\begin{document}
\makeatletter
    \def\blfootnote{\gdef\@thefnmark{}\@footnotetext}
    \let\@fnsymbol\@roman
    \makeatother

\maketitle

%%%%%%%%%%%%%%%%%%%%%%%%%%%%%%%%%%%%%%%%%%%%%%%%%%%%%%%%%%%%%%%
%%%%%%%%%%%%%%%%%%%%%%%%%%%%%%%%%%%%%%%%%%%%%%%%%%%%%%%%%%%%%%%
%%%%%%%%%%%%%%%%%%%%%%%%%%%%%%%%%%%%%%%%%%%%%%%%%%%%%%%%%%%%%%%

\begin{abstract}
In this paper we use a recent version of the Ruelle-Perron-Frobenius 
Theorem to compute, in terms of the maximal eigendata of the Ruelle
 operator, 
the pressure derivative of translation invariant spin systems
taking values on a general compact metric space. 
On this setting the absence of metastable states for continuous potentials
on one-dimensional one-sided lattice is proved.  
We apply our results, to show that the pressure 
of an essentially one-dimensional 
Heisenberg-type model, on the lattice $\mathbb{N}\times \mathbb{Z}$, 
is Fr\'echet differentiable, on a suitable Banach space. Additionally, 
exponential decay of the two-point function, for this model, 
is obtained for any positive temperature.
\end{abstract}

\blfootnote{\textup{2010} \textit{Mathematics Subject Classification}: 37D35.}

\section{Introduction}
The Ruelle operator  was introduced by David Ruelle in the seminal 
paper  \cite{Ruelle-1968}, in order to  prove 
the  existence and uniqueness of the Gibbs measures 
for some long-range Statistical Mechanics models in the 
one-dimensional lattice. Ever since  
the Ruelle operator  has become a standard tool 
in a variety  of mathematical fields, for instance, in Dynamical Systems, 
and other branches of Mathematics and Mathematical Physics.
The Ruelle operator was generalized in several 
directions and its generalizations are commonly 
called transfer operators.  Transfer operators 
appear in IFS theory, Harmonic Analysis and 
$C^{*}$-algebras, see for instance
 \cite{Excel-Lopes,Lau,Straub} respectively.

In Dynamical Systems the existence 
of Markov Partitions allows one to conjugate
uniform hyperbolic maps on compact  differentiable manifolds
with the shift map in the   Bernoulli space.
For more details see, for example, \cite{bo} and references therein.

A field in which the Ruelle operator formalism  has also been proved useful  
is  the Multifractal Analysis. Bowen, in the seminal work
\cite{boo}, has established a relationship between the  Hausdorff dimension 
of certain fractal sets and  topological pressure, 
for more
details see \cite{boo,ma} 
and also the introductory texts \cite{Barreira,Pesin}.

The classical Thermodynamic Formalism was  originally developed 
in the Bernoulli space $M^{\mathbb{N}}$, with $M$ being a finite alphabet, see \cite{PP}. 
The motivation to consider more general alphabets from the 
dynamical system point of view is given
in \cite{Sarig 1,Sarig 2}, where proposed models with 
infinite alphabet $M=\mathbb{N}$ are used to 
describe some  non-uniformly hyperbolic maps, 
for instance, the Manneville-Pomeau maps.
Unbounded alphabets as general standard Borel spaces, which includes 
compact and non-compact, are considered in details in \cite{Geogii88}.

In \cite{le} the authors considered the alphabet $M=S^1$ and a Ruelle operator formalism is developed. 
Subsequently, in \cite{LMMS}, this formalism was extended to general compact metric alphabets.
Those alphabets do not fit in the classical theory, since the number of preimages 
under the shift map may be not countable. To circumvent this 
problem the authors considered an a priori measure $\mu$
defined on $M$, and so a generalized Ruelle operator can be defined  
and a version of the Ruelle-Perron-Frobenius Theorem is proved. 
In this general setting concepts of entropy and pressure are also introduced. 
A variational principle is obtained in \cite{LMMS}. The authors also show that their 
theory can also be used to recovery some results of Thermodynamic Formalism for countable 
alphabets, by taking the one-point compactification of $\mathbb{N}$, 
and choosing a suitable a priori measure $\mu$.

In Classical Statistical Mechanics uncountable alphabets 
shows up, for example, in the so-called $O(n)$ models
with $n\geq 2$. These  are models 
on a $d$-dimensional lattice for which a vector 
in the $(n-1)$-dimensional  sphere is assigned   
to every lattice site  and the vectors  at 
adjacent sites interact  ferromagnetically 
via their inner product.
More specifically, let $n\geq 1$ be an integer  
and let $G=(V(G), E(G))$ a finite graph.  A configuration 
of the {\it Spin $O(n)$ model} on $G$ is an 
assignment $\sigma:V(G)\to S^{n-1}$, 
we denote by $\Omega:=(S^{n-1})^{V(G)}$  the 
space of configurations. At 
inverse temperature $\beta\in (0, \infty),$ 
configurations are randomly chosen from 
the probability measure $\mu_{G,n, \beta}$ 
given by 
\[
d\mu_{G,n, \beta}(\sigma):=\dfrac{1}{Z_{G, n, \beta}}
\exp ( \beta \sum_{u,v \in E(G)} \sigma_u\cdot \sigma_v) \ d\sigma,
\]
where $\sigma_u\cdot \sigma_v$ denotes the inner product 
in $\mathbb{R}^n$ of the assignments $\sigma_u$ and $\sigma_v$, 
$Z_{G, n, \beta}=\int_{\Omega}
\exp( 
\beta \sum_{u,v \in E(G)} \sigma_u\cdot \sigma_v)\, d\sigma
$
and $d\sigma$ is the uniform probability measure 
on $\Omega$.

Special cases of these models have names of their own:
 when $n=0$ this model is 
   the Self Avoiding Walking (SAW); when $n=1$, 
this model is  the Ising model,
when $n=2$,  the model is 
called the $XY$ model; finally  
for $n=3$ this model is called  the Heisenberg model,
see \cite{le, Ellis,GJ-Book,Geogii88, ONModels} for details.

In \cite{CL-SIL} the authors generalize the previous 
Ruelle-Perron-Frobenius Theorem for a more general class of potentials, 
satisfying what the authors called weak and  strong Walters conditions, 
which in turn are natural generalizations for the classical 
Walters condition.
The regularity properties of the pressure functional are studied 
and its analyticity is proved on the space of the Hölder 
continuous potentials. 
An exponential decay rate for correlations are obtained, in the case where 
the Ruelle operator exhibits the spectral gap property.
An example, derived from the long-range Ising model,
of a potential in the Walters class where 
the associated Ruelle operator has no
spectral gap is  given.

One of the main results of this work provide an 
explicit expression for the derivative 
of the pressure $P:C^{\alpha}(\Omega)\to \mathbb{R}$, 
where $\Omega \equiv M^{\mathbb{N}}$ and $M$ is a general compact metric space.
To be more precise, we show that  
\[\tag{1}\label{pressure-derivative-0}
P^{'}(f)\varphi=\int_{\Omega}\varphi h_f\, d\nu_f,
\]
where $h_f$ and $\nu_f$ are eigenfunction 
and  eigenmeasure of the associated Ruelle operator.

The proof follows closely the one given 
in \cite{ma}, where the expression \eqref{pressure-derivative-0} is obtained 
in the context of finite alphabet.
We also prove the existence of the limit
\[
P(f)=\lim_{n\to \infty}\dfrac{1}{n}\log \mathscr{L}_{f}^n{\bf 1}(x)
\]
in the uniform sense, for any continuous potential $f$. We would
like  to point out  that the existence of this limit in this setting has been  proved in 
\cite{CL16}. 
We give a new proof of this fact here for two reasons: 
first, it is different from the proof presented in \cite{CL16} 
and we believe it is more flexible to be adapted to other contexts;
second, some pieces of it are used to compute the pressure 
derivative. 

In the last section we apply our results in cases where $M=(S^2)^{\mathbb{Z}}$ 
endowed with a suitable a priori DLR-Gibbs measure. 
We introduce a Heisenberg-type model on the lattice $\mathbb{N}\times \mathbb{Z}$, 
depending on a real parameter $\alpha$, and 
use the Ruelle operator to obtain differentiability of the pressure and 
exponential decay rate for the two-point function for any positive temperature.

\section{Basic Definitions and Results}

In this section we set up the notation and present 
some preliminaries results. 
Let $M=(M,d)$ be a compact metric space, equipped 
with a Borel probability measure $\mu:\mathscr{B}(M)\to [0,1]$, 
having the whole space $M$ as its support. 
In this paper, the set of positive integers is denoted by $\mathbb{N}$.  
We shall denote by $\Omega=M^{\mathbb{N}}$ 
the set of all sequences $x=(x_1,x_2,\ldots )$, where 
$x_i \in M$, for all $i\in {\mathbb N}$. 
We denote the left shift mapping by $\sigma:\Omega\to\Omega,$ 
 which is given 
by $\sigma(x_1,x_2,\ldots)=(x_2,x_3,\ldots)$.
We consider the metric $d_{\Omega}$ on $\Omega$ given by
\[
d_{\Omega} (x,y)
=
\sum_{n=1}^{\infty} \frac{1}{2^n}d(x_n,y_n).
\]
The metric $d_{\Omega}$ induces the product topology
and therefore it follows from  Tychonoff's theorem that 
$(\Omega,d_{\Omega})$ is a compact metric space.
%We  will also consider the  space 
%$\hat{\Omega}=M^{\mathbb{Z}}$ of 
%the all sequences
% $x=(\ldots x_{-2},x_{-1}|x_0,x_1, \ldots )$
%endowed with the product metric 
%$
%d_{\hat{\Omega}}(x,y)
%=
%\sum_{n\in \mathbb{Z}}1/2^{|n|}d(x_n,y_n).
%$
%Denote by $\hat{\sigma}$ the left shift  on
%$\hat{\Omega}$. The space $\hat{\Omega}$ is 
%for the same reasons also a compact metric  space.

The space of all  continuous real functions $C(\Omega, \mathbb R)$
is denoted simply by $C(\Omega)$ and will be endowed with the norm
$\|\cdot\|_0$ defined  by
$\|f\|_0=\sup_{x\in \Omega} |f(x)|,$ which in turn is  a Banach space. 
For any fixed $0< \alpha< 1$ we denote by $C^{\alpha}(\Omega)$ 
the space of all $\alpha$-H\"older continuous functions, 
that is, the set of all functions $f:\Omega\to\mathbb{R}$ satisfying 
\[
\mathrm{Hol}_{\alpha}(f)
=
\sup_{x,y\in\Omega: x\neq y}
\dfrac{|f(x)-f(y)|}{d_{\Omega}(x,y)^{\alpha}}
<+\infty.
\] 

We equip $C^{\alpha}(\Omega),~0< \alpha< 1,$ with the norm 
 given by  
$\|f\|_{\alpha}= \|f\|_0+ \mathrm{Hol}(f)$.
We recall that 
$(C^{\alpha}(\Omega),\|\cdot\|_{\alpha}) $ is a Banach space for
any $0< \alpha< 1$.

Our \emph{potentials} will be elements of $C(\Omega)$
and, in order to have a well defined Ruelle operator 
when $(M,d)$ is a general compact metric space, we need 
to consider an {\it a priori measure} which is simply a 
Borel probability measure $\mu:\mathscr{B}(M)\to \mathbb{R}$,
where $\mathscr{B}(M)$ denotes the Borel $\sigma$-algebra 
of $M$. For many of the most popular choices of an uncountable
space $M$, there is a natural a priori measure $\mu$. 
Throughout this paper the a priori measure $\mu$ 
is supposed to have the whole space $M$ as its support. 
The Ruelle operator 
$\mathscr{L}_f: C(\Omega) \to C(\Omega)$
is the mapping sending $\varphi$ to $\mathscr{L}_{f}(\varphi)$ 
defined for any $x\in\Omega$ by the expression
\[
\mathscr{L}_f(\varphi)(x)
=
\int_M e^{f(ax)}\varphi(ax)d\mu(a),
\]
where $ax$ denotes the sequence 
$ax=(a, x_1, x_2, \ldots)\in \Omega$. 
%
%This operator is a generalization of the classical 
%Ruelle operator and has been appeared lately 
%in the Thermodynamic Formalism literature, see for example 
%\cite{le,LMMS}.
The classical Ruelle operator can be recovered on this  
setting by considering $M=\{0,1,\ldots,n\}$ and 
the  a priori measure  $\mu$ as the normalized counting measure.

\begin{theorem}[Ruelle-Perron-Frobenius]
\label{teo-RPF-compacto}
Let $(M,d)$ be a compact metric space, $\mu$ a 
Borel probability measure on $M$ having full support and 
$f$ a potential in $C^{\alpha}(\Omega)$, 
where $0<\alpha<1$. Then
$\mathscr{L}_f: C^{\alpha}(\Omega) \to C^{\alpha}(\Omega)$
has a simple positive eigenvalue of maximal modulus $\lambda_f$, 
and there are a strictly positive function $h_f$ 
and a Borel probability measure $\nu_{f}$ on $\Omega$ such that,

\begin{itemize}
	 
\item[(i)]
$\mathscr{L}_f h_f=\lambda_f h_f,$ $\mathscr{L}^{*}_f\nu_f=\lambda_f\nu_f$; 

\item[(ii)]  
the remainder of the spectrum of 
$\mathscr{L}_f: C^{\alpha}(\Omega) \to C^{\alpha}(\Omega)$ is con\-tai\-ned in 
a disc of radius strictly smaller than $\lambda_f$;

\item[(iii)]
for all continuous functions $\varphi\in C(\Omega)$ we have
\[
\lim_{n\to\infty}
\left\|
\lambda_{f}^{-n}\mathscr{L}^{n}_{f}\varphi-h_f\int \varphi\,  d\nu_{f}
\right\|_0
=
0.
\]
\end{itemize}
\end{theorem}

\begin{proof}
See \cite{le} for the case $M=S^1$ and \cite{LMMS} 
for a general compact metric space.
\end{proof}

%We say that a potential $f\in C(\Omega,\mathbb{R})$ is 
%{\bf normalized} if $\mathscr{L}_f({\bf 1})={\bf 1}.$ If 
%$f\in C^{\alpha}(\Omega, \mathbb{R})$ then there is a normalized potential 
%$\overline{f}=f-\log f\circ \sigma-\log h_f\circ\sigma-\log \lambda_f$, 
%which is cohomologous to $f$. 
%In particular, follows from RPF theorem that 
%\[
%\mathscr{L}_{\overline{f}}^n(\varphi) 
%\xrightarrow{\ n\to\infty\ } 
%\int_{\Omega} \varphi\, dm_f,
%\]
%where $m_f=h_f\nu_f$ is  a fixed point for the dual 
%$\mathscr{L}^{*}_{\overline{f}}.$

Following the references \cite{le,LMMS} we define 
the entropy of a shift invariant measure
and the pressure of the potential $f$, respectively, as follows
\[
h(\nu)
=
\inf_{f \in C^{\alpha}(\Omega)}
\left\{-\int_{\Omega}f d\nu+ \log\lambda_{f}  \right\}
\; \text{and} \;\,
P(f)
=
\sup_{\nu\in \mathscr{M}_{\sigma}}
\left\{h(\nu)+\int_{\Omega}f\, d\nu \right \},
\]
where $\mathscr{M}_{\sigma}$ is the set of all shift invariant Borel 
probability measures. 

\begin{proposition}[Variational Principle]\label{Principio variacional}
For each $f \in C^{\alpha}(\Omega)$ we have for all 
$x\in \Omega$ that 
\[
P(f)
=
\lim_{n\to\infty}
\frac{1}{n}
\log[ \mathscr{L}_{f}^n({\bf 1})(x)]
=
\log\lambda_{f}
=\sup_{\nu\in \mathscr{M}_{\sigma}}
\left\{h(\nu)+\int_{\Omega}f\, d\nu \right \}.
\]
Moreover the  supremum is attained by $m_f=h_f \nu_f$.  
\end{proposition}
\begin{proof}
See \cite{LMMS} Corollary 1.
\end{proof}

\begin{theorem}\label{Pressure Analiticity}
The function defined by $C^{\alpha}(\Omega)\ni f\mapsto P(f)\in \mathbb{R}$ 
is a real analytic function.
\end{theorem}

\begin{proof}
See \cite{CL-SIL} for a proof.
\end{proof}

Let $f\in C^{\alpha}(\Omega)$ be a Hölder continuous potential, $\nu_f$
be the measure given by Theorem \ref{teo-RPF-compacto} and
$\varphi, \psi\in C^{\alpha}(\Omega)$. For each $n\in \mathbb{Z}$
we define the correlation function
\begin{equation}\label{correlation-function}
C_{\varphi,\psi,m_f}(n)
=
\int (\varphi\circ \sigma^n)\psi \, dm_f
-\int \varphi \,dm_f \int \psi\, dm_f.
\end{equation}
We have that the above correlation function  has exponential decay. 
More precisely we have the following proposition:
\begin{proposition}\label{Decay Correlation}
For each $n\in \mathbb{N}$ let $C(n)$ denote the correlation function 
defined by \eqref{correlation-function}. Then there exist constants
$K>0$ and $0<\tau<1$ such that $|C(n)|\leq K \tau^n$.
\end{proposition}

The proof when $M$ is finite is given in \cite{Baladi}. Due to 
Theorem \ref{teo-RPF-compacto} this proof can be easily adapted to
the case where $M$ is compact metric space. 
We will include the details here for completeness.
Before prove the above theorem, we present two auxiliary lemmas.
\begin{lemma}\label{lema-proj-spectral-pif}
Let $f\in C^{\alpha}(\Omega)$, and $\partial D$ the boundary of a disc $D$ with  center in $\lambda_f,$
then the spectral projection 
\[
\pi_{f}\equiv \pi_{\mathscr{L}_{f}}=\int_{\partial D}(\lambda I-\mathscr{L}_{f})^{-1}d\lambda.
\] 
 is given by 
$
\pi_{f}(\varphi) 
= 
\big(\int \varphi\, d\nu_{f}\big) \cdot h_{f}.
$
\end{lemma}

\begin{lemma}\label{lema-phi1-fora}
Let be $\varphi,\psi \in C^{\alpha}(\Omega)$ then
$
\mathscr{L}_{f}^n(\varphi\circ \sigma^n \cdot \psi \cdot h_{f})
=
\varphi \mathscr{L}_{f}^n( \psi h_{f}). 
$
\end{lemma}
The proof of both lemmas are 
straightforward computation so they will be omitted.

\bigskip 
\noindent
{\bf Proof of  Proposition \ref{Decay Correlation}}.
Since $m_f=h_{f} d\nu_{f}$ it follows from  
the definition of the correlation function that
\begin{align*}\label{D. Correlation 2}
|C_{\varphi,\psi,m_f }(n)|
&=
\left| 
\int (\varphi\circ \sigma^n)\psi h_{f}\, d\nu_{f}
- 
\int\varphi  h_{f}\, d\nu_{f}
\int\psi h_{f}\, d\nu_{f}                            
\right|.
\end{align*}
Notice that 
$
(\mathscr{L}^{*}_{f})^n\nu_{f}
=
\lambda_{f}^n\nu_{f}
$ 
and therefore the rhs above is equal to 
\[
\left|
\int\lambda_{f}^{-n}
\mathscr{L}_{f}^n((\varphi\circ \sigma^n)\psi h_{f})
\, d\nu_{f}
-
\int\varphi  h_{f}\, d\nu_{f}
\int\psi h_{f}\,  d\nu_{f}                            
\right|.
\]
By using the Lemma \ref{lema-phi1-fora} and performing simple 
algebraic computations we get 
\begin{align} \label{estimativa-1-dec-pol}
|C_{\varphi,\psi,m_f }(n)|
\leq 
\left(   \int|\varphi|\, d\nu_{f}    \right)
\left\|
\lambda^{-n}_{f} \mathscr{L}_{f}^n
\left(
	\psi h_{f}-  h_{f}\int\psi h_{f}\, d\nu_{f}  
\right) 
\right\|_{0}.
\end{align}

By Theorem \ref{teo-RPF-compacto} we know that the spectrum of 
$\mathscr{L}_{f}: C^{\alpha}(\Omega)\to C^{\alpha}(\Omega)$ 
consists in a simple eigenvalue $\lambda_{f}>0$ and a subset 
of a disc of radius strictly smaller than $\lambda_{f}.$  
Set 
$
\tau=\sup\{|z|; |z|<1~\textnormal{and}~z\cdot\lambda_{f} \in 
\mathrm{Spec}(\mathscr{L}_{f})\}.
$ 
The existence of the spectral gap guarantees that $\tau<1$. 
Let $\pi_{f}$ 
the spectral projection associated to eigenvalue $\lambda_{f}$,  
then, 
the spectral radius of the operator 
$\mathscr{L}_{f}(I-\pi_{f})$ is exactly $\tau\cdot \lambda_f$. 
Since the commutator $[\mathscr{L}_{f},\pi_{f}]=0$,
we get $\forall n \in\mathbb{N}$ that
$
[\mathscr{L}_{f}(I-\pi_{f})]^n
=
\mathscr{L}_{f}^{n}(I-\pi_{f}).
$
From the spectral radius formula  
it follows that for each choice of $\widetilde{\tau}>\tau$ 
there is $n_0\equiv n_0(\widetilde{\tau})\in\mathbb{N}$ 
so that for all $n\geq n_0$ we have
$
\|\mathscr{L}_{f}^{n}(\varphi-\pi_{f}\varphi)\|_0
\leq 
\lambda_{f}^n\widetilde{\tau}^n
\|\varphi\|_0,~\forall\varphi\in  C^{\alpha}(\Omega).
$
Therefore there is a constant 
$C( \widetilde{\tau})>0$ 
such that for every $n\geq 1$
\[ 
\|\mathscr{L}_{f}^n(\varphi-\pi_{f}\varphi)\|_0
\leq 
C( \widetilde{\tau}) \ \lambda_{f}^n\ \widetilde{\tau}^n\ \|\varphi\|_0
\qquad\forall\varphi\in  C^{\alpha}(\Omega).
\]
By using the Lemma \ref{lema-proj-spectral-pif} 
and the above upper bound 
in the inequality \eqref{estimativa-1-dec-pol}
we obtain 
\begin{align*}
|C_{\varphi,\psi,m_f}(n)|
&\leq
\left( \int |\varphi|\, d\nu_{f} \right)
C\ \widetilde{\tau}^n\ \|\psi h_{f}\|_0
\\
&\leq
C(\widetilde{\tau}) 
\|h_{f}\|_{0} 
\left( \int |\varphi|d\nu_{f} \right)
\|\psi\|_0 \ \widetilde{\tau}^n.
\tag*{\qed} 
\end{align*}

\section{Main Results}
%%%%%%%%%%%%%%%%%%%%%%%%%%%%%%%%%%%%%%%%%%%%%%%%%%%%%%%%%%%%%%%%%%%%%%%
%%%%%%%%%%%%%%%%%%%%%%%%%%%%%%%%%%%%%%%%%%%%%%%%%%%%%%%%%%%%%%%%%%%%%%%
%%%%%%%%%%%%%%%%%%%%%%%%%%%%%%%%%%%%%%%%%%%%%%%%%%%%%%%%%%%%%%%%%%%%%%%
%%%%%%%%%%%%%%%%%%%%%%%%%%%%%%%%%%%%%%%%%%%%%%%%%%%%%%%%%%%%%%%%%%%%%%%
Proposition \ref{Principio variacional} ensures for any H\"older potential $f$
that the limit
$P(f)
=
\lim_{n\to\infty}
n^{-1}
\log[ \mathscr{L}_{f}^n({\bf 1})(x)]
$
always exist and is independent of $x\in \Omega$. 
In what follows, we will extend this result for all continous potentials
$f\in C(\Omega)$. This is indeed a surprising result and it does not have a
counter part on one-dimensional two-sided lattices, 
due to the existence of metastable states discovered by Sewell in \cite{Sewell}.
Before present the proof of this fact, 
we want to explain what is the mechanism behind the absence of metastable 
states for one-dimensional systems on the lattice $\mathbb{N}$.
The absence of metastable states for continuous potentials and finite state
space $M$, on one-dimensional one-sided lattices, as far as we known 
was first proved by Ricardo  Ma\~n\'e in \cite{ma}, but apparently he did not realize it.
A generalization of this result for $M$ being a compact metric space appears 
in \cite{CL16} and again no mention to metastable states is made in this paper. 
The proof of this result presented in \cite{CL16} is completely different from ours and we
believe that one presented here is more suitable to be adapted to other context. 
An alternative explanation of this fact can be found in \cite{CER17} Remark 2.2,
and following the first author of \cite{CER17} the first person to realize this
fact was Aernout van Enter. 

Let $(\widetilde{\mathscr{B}},\|\cdot\|)$ and $(\mathscr{B},|\!|\!|\cdot|\!|\!| )$ the classical 
Banach spaces of interactions defined as in \cite{Israel}. 
In case of free boundary conditions, we have for any $\Phi\in\widetilde{\mathscr{B}}$
and a finite volume $\Lambda_n\subset\mathbb{Z}$ that
$
|P_{\Lambda_n}(\Phi)-P_{\Lambda_n}(\Psi)| \leq |\!|\!|\Phi-\Psi|\!|\!|
$
and therefore the finite volume pressure with {\bf free boundary conditions} 
is 1-Lipschitz function 
from $(\widetilde{\mathscr{B}},|\!|\!|\cdot|\!|\!|)$ to $\mathbb{R}$. 
On the lattice $\mathbb{Z}$, when boundary conditions are considered 
the best we can prove is 
$
|\tau_nP_{\Lambda_n}(\Phi)-\tau_nP_{\Lambda_n}(\Psi)|
\leq 
\|\Phi-\Psi\|,
$
for interactions $\Phi,\Psi\in \mathscr{B}$.
From this we have that finite volume pressure with boundary conditions 
is 1-Lipschitz function from the smaller Banach space 
$(\mathscr{B},\|\cdot\|)$ to $\mathbb{R}$. The last inequality can not be 
improved for general boundary conditions and interactions in the big Banach space
$(\widetilde{\mathscr{B}},|\!|\!|\cdot|\!|\!|)$, because 
as we will see in the proof of Theorem \ref{Pressao-assintotica} 
it would imply the independence of the boundary 
conditions of the infinite volume pressure which is a contradiction with 
Sewell's theorem. The mechanism that prevents existence of metastable 
states for continuous potential on the lattice $\mathbb{N}$ is the 
possibility of proving that the analogous of the finite volume pressure
with boundary conditions on the lattice $\mathbb{N}$ is indeed 1-Lipschitz
function. To be more precise.

\begin{theorem}\label{Pressao-assintotica}
For each continuous potential  $f\in C(\Omega)$ there is a real number $P(f)$ such that 
$$
\lim_{n\to \infty}
\left\| \dfrac{1}{n}\log\mathscr{L}_f^n{\bf 1} -P(f)\right\|_{0}
=0.
$$
\end{theorem}

\begin{proof} 
It is sufficient to prove that  
$\Phi_{n}:C(\Omega) \to C(\Omega)$, given by 
$$ 
\Phi_{n}(f)(x)=\frac{1}{n}\log \mathscr{L}_f^n{\bf 1}(x)
$$
converges to a Lipschitz continuous function $\Phi:C(\Omega) \to C(\Omega)$ 
in the following sense $\|\Phi_n(f)-\Phi(f)\|_{0}\to 0$, when $n\to\infty$.
Indeed, by Proposition \ref{Principio variacional}, for any fixed $0<\alpha< 1$ we have   
$\Phi(C^{\alpha}(\Omega))\subset \langle 1\rangle$,
where $\langle 1 \rangle$ denotes the subspace generated by the constant functions
in $C(\Omega)$. Since $C^{\alpha}(\Omega)$ is a dense subset of $(C(\Omega),\|\cdot\|_{0})$
and $\Phi$ is Lipschitz, then 
$\Phi(C(\Omega))=\Phi(\overline{C^{\alpha}(\Omega)})\subset \langle 1 \rangle$.

\bigskip

In order to deduce the convergence of
$
(\Phi_{n})_{n\in \mathbb{N}},
$  
it is more convenient to identify 
$
\Phi_{n}:C(\Omega)\to C(\Omega)
$
with the function  
$
\Phi_{n}:C(\Omega)\times \Omega\rightarrow\mathbb{R},
$
given by  
$
\Phi_{n}(f,x)=(1/n)\log \mathscr{L}_f^n{\bf 1}(x).
$ 
%%
%%%

For any fixed $x\in\Omega$, 
follows from the Dominated Convergence Theorem
that the Fr\'echet derivative of 
$
\Phi_{n}:C(\Omega)\times \Omega\rightarrow\mathbb{R},
$ 
evaluated  at $f$ and computed  in $\varphi$, is given by  
$$
\frac{\partial}{\partial f}\Phi_n(f,x)\cdot\varphi
=
\frac{1}{n}
\frac{\displaystyle\int_{M^n}(S_{n}\varphi)(\ba x)\exp(S_{n}f)(\ba x)\, d\mu(\ba)}
{\mathscr{L}_f^n{\bf 1}(x)},
$$
where $(S_{n}\varphi)(x)\equiv  \sum_{j=0}^{n-1}\varphi\circ\sigma^{j}(x).$
Clearly,  
$
\displaystyle\left\| n^{-1}S_n\varphi\right\|_{0}
\leq
\|\varphi\|_{0},
$
and therefore we have the following estimate
\begin{align}\label{pressao 3}
\left|\frac{\partial}{\partial f}\Phi_{n}(f,x)\cdot\varphi\right|
&=
\left|\frac{\displaystyle\int_{M^n}\frac{1}{n}(S_{n}\varphi)(\ba x)\exp(S_{n} f)(\ba x)\, d\mu(\ba)}
{\mathscr{L}_f^n{\bf 1}(x)}\right|
\nonumber
\\
&\leq
\frac{\displaystyle\int_{M^n}\left|\frac{1}{n}(S_{n}\varphi)(\ba x)\exp(S_{n}f)(\ba x)\, \right| d\mu(\ba)}
{\left|\displaystyle\int_{M^n}\exp(S_{n}f)(\ba x)\, d\mu(\ba)\right|}
\nonumber
\\
&\leq 
\|\varphi\|_0,
\end{align}
for any $f\in C(\Omega)$ independently of $x\in \Omega$.
Taking the supremum in \eqref{pressao 3} over $x\in\Omega$ we obtain   
\[
\left\|\frac{\partial}{\partial f}\Phi_{n}(f,\cdot)\cdot\varphi\right\|_0
\leq 
\|\varphi\|_0,
\]
for all  $n\in\mathbb{N}$. The above inequality 
allows to conclude  that 
$\|\frac{\partial}{\partial f}\Phi_{n}(f)\| \leq 1$, 
where $\|\cdot\|$ means the operator norm.

Fix $f$ and $\tilde{f}$ in $C(\Omega)$ 
and define for each  $n\in \mathbb{N}$ the map  
$\hat{\Phi}_{n}(t)=\Phi_{n}(\alpha(t),x),$ 
where
$\alpha(t)=t f+(1-t)\tilde{f}$
with   $0\leq t\leq 1$. Obviously, $\hat{\Phi}_{n}$  
is a differentiable map when  seen as 
a map from $[0,1]$ to  $\mathbb{R}$ and we have that 
$
|\hat{\Phi}_{n}(1)-\hat{\Phi}_{n}(0)|=|\frac{d}{dt}\hat{\Phi}_{n}(\hat{t})(1-0)|
$ 
for some 
$\hat{t}\in (0,1)$. 
Using the above estimative of the Fr\'echet derivative norm
we have that 
\begin{align}\label{estimativa-dif-Phin}
|\Phi_{n}(f,x)-\Phi_{n}(\tilde{f},x)|
=
\left|\frac{\partial}{\partial f}\Phi_{n}(f,x)(f-\tilde{f})\right|\leq\|f-\tilde{f}\|_{0}.
\end{align}
As an outcome for any fixed $f\in C(\Omega)$ the sequence  $(\Phi_{n}(f))_{n\in\mathbb{N}}$ 
is  uniformly equicontinuous.
Moreover, $\sup_{n\in\mathbb{N}}\|\Phi_{n}(f)\|_{0}<\infty$. 
Indeed, from inequality $\eqref{estimativa-dif-Phin}$, the triangular inequality 
and existence of the limit  $\lim_{n\to\infty} \Phi_n({\bf 1})$,  
it follows that
\begin{align*}
 |\Phi_n(f)|
 =
 |\Phi_n(f)-\Phi_n({\bf 1})|+|\Phi_n({\bf 1})|
 &\leq 
 \|f-{\bf 1} \|_{\infty} +|\Phi_n({\bf 1})|
 \\
 &\leq
 \|f-{\bf 1}\|_{\infty} +\sup_{n\in\mathbb{N}}|\Phi_n({\bf 1})|
 \\
 &\equiv 
 M(f).
 \end{align*} 
Now, we are able  to apply the  Arzel\`a-Ascoli's Theorem 
to obtain a  subsequence
$
(\Phi_{n_k}(f))_{k\in\mathbb{N}},
$ 
which  converges to a function $\Phi(f)\in C(\Omega)$.

 %A simple $``\varepsilon/3"$ argument 
We now show that   $\Phi_n(f)\to \Phi(f)$, when $n\to\infty$.
Let $\varepsilon>0$ 
and  $g\in C^{\alpha}(\Omega)$ such that $\|f-g\|_0<\varepsilon$.
Choose $n_k$ and $n$ sufficiently large so that  
the inequalities  
$\|\Phi_{n_k}(f)- \Phi(f)\|_{0}<\varepsilon$,
$\|\Phi_n(g)-\Phi_{n_k}(g)\|_{0}<\varepsilon$
and
$\|\Phi_{n_k}(g)-\Phi_{n_k}(f)\|_{0}<\varepsilon$
are satisfied.
For these choices of  $g,n$ and  $n_k$, we have by the triangular inequality
and inequality \eqref{estimativa-dif-Phin} that
\begin{align*}
\|\Phi_n(f)- \Phi(f)\|_{0}
&\leq 
\|\Phi_n(f)-\Phi_{n_k}(f)\|_{0}
+
\|\Phi_{n_k}(f)- \Phi(f)\|_{0}
\\
&<
\|\Phi_n(f)-\Phi_{n_k}(f)\|_{0}
+
\varepsilon
\\
&\leq 
\|\Phi_n(f)-\Phi_{n}(g)\|_{0}
+
\|\Phi_n(g)-\Phi_{n_k}(f)\|_{0}
+
\varepsilon
\\
&< 
\|\Phi_n(g)-\Phi_{n_k}(f)\|_{0}
+
2\varepsilon
\\
&\leq 
\|\Phi_n(g)-\Phi_{n_k}(g)\|_{0}
+
\|\Phi_{n_k}(g)-\Phi_{n_k}(f)\|_{0}
+
2\varepsilon
\\
&<
4\varepsilon,
\end{align*}
thus proving the desired convergence. 

To finish the proof it is enough to observe 
that the inequality \ref{estimativa-dif-Phin}
implies that $\Phi$ is a Lipschitz continuous function.
\end{proof}

\begin{theorem}\label{Pressure derivative}
For each fixed  $0<\alpha<1$ and $f\in C^{\alpha}(\Omega)$ 
the  Fr\'echet derivative of the pressure  functional 
$P:C^{\alpha}(\Omega)\rightarrow\mathbb{R}$
is given by
\begin{equation}\label{pressure-derivative}
P'(f)\varphi=\int\varphi h_{f}\, d\nu_{f} 
\end{equation}
\end{theorem}

\begin{lemma}
For each  fixed $f \in  C^{\alpha}(\Omega),$ $0<\alpha< 1,$ there exists 
the limit
$$
\lim_{n\to \infty}
\left\|
\frac{1}{n}
\frac{\mathscr{L}_{f}^{n}(S_{n}\varphi)}
{\mathscr{L}^{n}_{f}{\bf 1}}
-
\int\varphi h_{f}\, d\nu_{f}
\right\|_0
=
0
$$
for every  $\varphi\in C(\Omega)$, and  the convergence is uniform.
\end{lemma}
\begin{proof} 
A straightforward calculation shows that   
\begin{align}
\label{desigualdade}
\left\|
\frac{1}{n}\frac{\mathscr{L}_{f}^{n}(S_{n}\varphi)}{\mathscr{L}^{n}_{f}{\bf 1}}
-
\int\varphi h_{f}\, d\nu_{f}
\right\|_{0}
&=
\left\|
\frac{\lambda_f^{n}}{\mathscr{L}^{n}_{f}{\bf 1}}\frac{1}{n}
\lambda_f^{-n}\mathscr{L}_{f}^{n}(S_{n}\varphi)
-
\int\varphi h_{f}\, d\nu_{f}
\right\|_{0}
\\[0.4cm] 
&
\hspace*{-3,175cm}
\leq
\sup_{n\in\mathbb{N}}
\left\|\frac{\lambda_f^{n}}{\mathscr{L}^{n}_{f}{\bf 1}}\right\|_{0}
\left\|\frac{1}{n}\lambda_f^{-n}\mathscr{L}_{f}^{n}(S_{n}\varphi)-
(\lambda_f^{-n}\mathscr{L}^{n}_{f}{\bf 1})\int\varphi h_{f}\, d\nu_{f}\right\|_{0}.
\nonumber
\end{align}
Therefore to get the desired result it is 
sufficient to show that:
\begin{itemize}
\item[(a)] 
$
\displaystyle\sup_{n}\left\|\frac{\lambda_f^{n}}{\mathscr{L}^{n}_{f}{\bf 1}}\right\|_{0} 
$ is finite;

\item[(b)]
$
\displaystyle \quad (\lambda_f^{-n}\mathscr{L}^{n}_{f}{\bf 1})\int\varphi h_{f}\, d\nu_{f}
$
converges to  
$
\displaystyle h_f\int\varphi h_{f}\, d\nu_{f};
$

\item[(c)] 
$
\displaystyle \quad\frac{1}{n}\lambda_f^{-n}\mathscr{L}_{f}^{n}(S_{n}\varphi) $
converges to
$
\displaystyle h_f\int\varphi h_{f}\, d\nu_{f}.
$
\end{itemize}

The first two items 
are immediate consequences of the Ruelle-Perron-Frobenius Theorem. Indeed,  the 
convergence 
$
\lambda_f^{-n}\mathscr{L}_f^n {\bf 1}\stackrel{\|\cdot\|_0}{\longrightarrow}h_f,
$
 immediately  give that
\begin{align*}
(\lambda_f^{-n}\mathscr{L}^{n}_{f}{\bf 1})
\int\varphi h_{f}\, d\nu_{f} 
\stackrel{\|\cdot\|_0}{\longrightarrow}  
h_f\int\varphi h_{f}\, d\nu_{f}.
\end{align*}
Since  $h_f$ is a continuous strictly positive function, 
it follows from compactness of $\Omega$ that  $h_f$ is bounded away  from  
zero, 
and consequently
 $
 \lambda_f^{n}/\mathscr{L}_f^n{\bf 1} \stackrel{\|\cdot\|_0}{\longrightarrow} 1/h_f.
 $
Once  $h_f$  is strictly positive $1/h_f$ is also positive and 
bounded away from zero, which gives  that  
\begin{align*}
\sup_{n}\left\|\frac{\lambda_f^{n}}{\mathscr{L}^{n}_{f}{\bf 1}}\right\|_{0}<\infty.
\end{align*}

The third expression  in  (c)  is  harder to analyze
than the previous  two, so  we will split  the analysis 
in three claims.

\medskip
\noindent \emph{Claim 1.} 
For all $\varphi\in C(\Omega)$ and   $n\in\mathbb{N}$ we have 
that,
\begin{equation}\label{Claim 1}
\lambda_f^{-n}\mathscr{L}_{f}^{n}(S_{n}\varphi)
=
\sum_{j=0}^{n-1}\lambda_f^{-(n-j)}\mathscr{L}_{f}^{n-j}
(\varphi\lambda_f^{-j}\mathscr{L}^{j}_{f}{\bf 1}).
\end{equation}
 
%\[
%\mathscr{L}_{f}^{n}(S_{n}\varphi)
%=
%\sum_{j=0}^{n-1}\mathscr{L}_{f}^{n-j}(\varphi\mathscr{L}^{j}_{f}{\bf 1}).
%\]
We first observe that,
$
\mathscr{L}_{f}^{n}(\varphi\circ\sigma^{n})
=
\varphi\mathscr{L}^{n}_{f}{\bf 1},
$ 
which  is an easy consequence of the definition of the  Ruelle operator.  
From that and  the linearity of the Ruelle operator it follows,
\begin{align*}
\mathscr{L}_{f}^{n}(S_{n}\varphi)
&=
\mathscr{L}_{f}^{n}(\varphi)+\mathscr{L}_{f}^{n-1}(\mathscr{L}_{f}(\varphi\circ\sigma))+
\ldots+\mathscr{L}_{f}(\mathscr{L}_{f}^{n-1}\varphi\circ\sigma^{n-1})
\\
&
\hspace*{-1cm}
=
\mathscr{L}_{f}^{n}(\varphi)+\mathscr{L}_{f}^{n-1}(\varphi\mathscr{L}_{f}{\bf 1})+
\ldots+\mathscr{L}_{f}(\varphi\mathscr{L}^{n-1}_{f}{\bf 1})
=
\sum_{j=0}^{n-1}\mathscr{L}_{f}^{n-j}(\varphi\mathscr{L}^{j}_{f}{\bf 1})
\end{align*}
%\begin{align*}
%\mathscr{L}_{f}^{n}(S_{n}\varphi)
%&=
%\mathscr{L}_{f}^{n-1}[\mathscr{L}_{f}\{\varphi+\varphi\circ\sigma+\varphi\circ\sigma^{2}
%+\ldots+\varphi\circ\sigma^{n-1}\}]
%\\
%&=
%\mathscr{L}_{f}^{n}(\varphi)+\mathscr{L}_{f}^{n-1}(\mathscr{L}_{f}(\varphi\circ\sigma))+
%\mathscr{L}_{f}^{n-2}(\mathscr{L}_{f}^{2}\varphi\circ\sigma^{2})+\ldots+\mathscr{L}_{f}(\mathscr{L}_{f}^{n-1}\varphi\circ\sigma^{n-1})
%\\
%&=
%\mathscr{L}_{f}^{n}(\varphi)+\mathscr{L}_{f}^{n-1}(\varphi\mathscr{L}_{f}{\bf 1})+
%\mathscr{L}_{f}^{n-2}(\varphi\mathscr{L}^{2}_{f}{\bf 1})+\ldots+\mathscr{L}_{f}(\varphi\mathscr{L}^{n-1}_{f}{\bf 1})
%\\
%&=
%\sum_{j=0}^{n-1}\mathscr{L}_{f}^{n-j}(\varphi\mathscr{L}^{j}_{f}{\bf 1})
%\end{align*}
finishing the proof.  
From  the linearity of the Ruelle operator 
we easily get  
\begin{align*}
\lambda_f^{-n}\mathscr{L}_{f}^{n}(S_{n}\varphi)
=
\sum_{j=0}^{n-1}\lambda_f^{-(n-j)}\mathscr{L}_{f}^{n-j}
(\varphi\lambda_f^{-j}\mathscr{L}^{j}_{f}{\bf 1}).
\end{align*}
 
\medskip
\noindent\emph{Claim  2.} For each  $\varphi\in C(\Omega)$ 
we have 
\begin{equation}\label{eq:5}
\lim_{n\to \infty}\left\|
\frac{1}{n}\left(
\lambda_f^{-n}\mathscr{L}_{f}^{n}(S_{n}\varphi)-\sum_{j=0}^{n-1}\lambda_f^{-(n-j)}\mathscr{L}_{f}^{n-j}\varphi h_{f}
\right)
\right\|_{0}
=
0.
\end{equation}
To verify \eqref{eq:5}  we use  \eqref{Claim 1}  to obtain the following estimate,
\begin{align*}
\left\|
	\frac{1}{n}
	\left\{
		\lambda_f^{-n}\mathscr{L}_f^{n}(S_{n}\varphi)
		-
		\sum_{j=0}^{n-1}\lambda_f^{-(n-j)}\mathscr{L}_f^{(n-j)}\varphi h_f
	\right\}
\right\|_{0}
%&
%\\
%\hspace*{-4cm}=
%\left\|
%	\frac{1}{n}
%	\left\{
%		\sum_{j=0}^{n-1}\lambda_f^{-(n-j)}\mathscr{L}_{f}^{n-j}(\varphi\lambda_f^{-j}\mathscr{L}^{j}_{f}{\bf 1})
%		-
%		\sum_{j=0}^{n-1}\lambda_f^{-(n-j)}\mathscr{L}_{f}^{n-j}\varphi h_f
%	\right\}
%\right\|_{0}
\\[0.4cm]
&\hspace*{-4cm}=
\left\|\frac{1}{n}\sum_{j=0}^{n-1}\lambda_f^{-(n-j)}\mathscr{L}_f^{n-j}\left(\varphi\lambda_f^{-j}\mathscr{L}^{j}_f{\bf 1}-\varphi h_f\right)\right\|_{0}
\\[0.4cm]
%&\hspace*{-4cm}\leq
%\frac{1}{n}\sum_{j=0}^{n-1}\left\|\lambda_f^{-(n-j)}\mathscr{L}_f^{n-j}\left(\varphi\lambda_f^{-j}\mathscr{L}^{j}_f{\bf 1}-\varphi %h_f\right)\right\|_{0}
%\\
%&\hspace*{-4cm}\leq
%\frac{1}{n}\sum_{j=0}^{n-1}\left\|
%\lambda_f^{-(n-j)}\mathscr{L}_f^{n-j}\right\|_{\mathscr{L}(C(\Omega),C(\Omega))}
%\left\|\varphi\lambda_f^{-j}\mathscr{L}^{j}_f{\bf 1}-\varphi h_f\right\|_{0}
%\\
&\hspace*{-4cm}\leq
\frac{const.}{n}\, 
\sum_{j=0}^{n-1}\left\|\varphi\lambda_f^{-j}\mathscr{L}^{j}_f{\bf 1}-\varphi h_f\right\|_{0}.
\end{align*}
The last term in the above inequality converges 
to zero, when $n\to \infty$, because it  is a Ces\`aro summation associated to 
the sequence  $\varphi\lambda_f^{-j}\mathscr{L}^{j}_f{\bf 1}-\varphi h_f$,
which converges to  zero in the uniform norm by    the 
Ruelle-Perron-Frobenius Theorem, so the claim is proved.

\noindent \emph{Claim 3.} For each  $\varphi\in C(\Omega)$ there exists  
the following limit
\begin{equation}\label{eq:7}
\lim_{n\to \infty}\left\|\frac{1}{n}\sum_{j=0}^{n-1}\lambda_{f}^{-(n-j)}\mathscr{L}_{f}^{n-j}\varphi h_{f}-
h_{f}\int\varphi h_{f}d\nu_{f}\right\|_{0}=0.
\end{equation}
Define
$A_{n,j}:=\lambda_f^{-(n-j)}\mathscr{L}_{f}^{n-j}\varphi h_{f}$ 
and 
$B:=h_{f}\int\varphi h_{f}\, d\nu_{f}.$ 
In one hand,
 by  the  Ruelle-Perron-Frobenius Theorem,
 we must have for any
fixed  
$j\in\mathbb{N}$ that \linebreak $\lim_{n\to \infty}\left\|A_{n,j}-B\right\|_{0}=0.$
On the other hand, we have by the triangular inequality 
and the convergence in the Ces\`aro sense that 
\begin{align*}
\left\|\dfrac{1}{n}\sum_{j=0}^{n-1}A_{n,j}-B\right\|_{0}
\leq
\dfrac{1}{n}\sum_{j=0}^{n-1}\left\|A_{n,j}-B\right\|_{0}\longrightarrow 0,
\end{align*}
when  $n\rightarrow\infty,$   finishing the proof of  Claim 3.
Therefore the proof of the Lemma is established.
\end{proof}

\noindent{\bf Proof of Theorem \ref{Pressure derivative}} 

 Fix $x\in \Omega $ and define  function
 $\Phi_{n}:C^{\alpha}(\Omega)\rightarrow\mathbb{R}$ 
 as %we define in the proof of Theorem  \ref{Pressao-assintotica},  
\[
\Phi_{n}(f)\equiv \frac{1}{n}\log(\mathscr{L}_{f}^{n}{\bf 1})(x).
\]
%As  we seen in Theorem  \ref{Pressao-assintotica},
As we have seen in the proof of Theorem \ref{Pressao-assintotica}  
the  Fr\'echet derivative
of  $\Phi_n$ at $f$ evaluated in  $\varphi\in C^{\alpha}(\Omega)$ is 
given  by  
\[
\Phi_{n}'(f)\varphi
=
\frac{1}{n}
\frac{\mathscr{L}_{f}^{n}(S_{n}\varphi)}{\mathscr{L}^{n}_{f}{\bf 1}},
\]
Since  we have  the analyticity of the pressure
functional in  $C^{\alpha}(\Omega)$ (Theorem \ref{Pressure Analiticity}) it follows from 
the previous  Lemma that 
\[
P'(f)\varphi
=
\lim_{n\to \infty} \Phi_{n}'(f)\varphi
=
\int\varphi h_{f}\, d\nu_{f}.
\]

\section{A Heisenberg type Model }
The aim of this section is to  introduce a Heisenberg type model on
the half-space $\mathbb{N}\times \mathbb{Z}$, 
prove absence of phase transition and exponential decay of correlations 
for this model.

The construction of this model is split in two steps. 
Firts step. We consider  $(S^2)^\mathbb{Z}$ as the configuration space. At 
inverse temperature $\beta\in (0, \infty),$ 
the configurations  are randomly chosen according 
to the following probabilities 
  measures $\mu_{n, \beta}$ 
\begin{equation}\label{Heisenberg-measure}
d\mu_{n, \beta}(\sigma):=\dfrac{1}{Z_{n, \beta}}
\exp (
\beta \sum_{i,j\in \Lambda_n} \sigma_i\cdot \sigma_j
)\, d\sigma,
\end{equation}
where $\Lambda_n$ denotes the symmetric interval of 
integers $[-n,n]$, $\sigma_i\cdot \sigma_j$ denotes the inner product 
 in $\mathbb{R}^3$ of  the first neighbors $\sigma_i$ and $\sigma_j$, 
\[
Z_{n, \beta}=\int_{(S^2)^{\Lambda_n}}\exp
(
\beta \sum_{i,j} {\sigma_i\cdot \sigma_j }
)\, d\sigma
\]
and $d\sigma$ is the uniform probability measure 
on $(S^{2})^{\Lambda_n}$.  Let measure $\hat{\nu}$ be 
the unique accumulation point of the sequence of probability measures
given by \eqref{Heisenberg-measure}, see \cite{Geogii88} for details.
The measure $\hat{\nu}$ will be used as the a priori measure 
in the second step.

Second step. Now we introduce a Heisenberg type model. 
We begin with the compact metric space $(S^2)^{\mathbb{Z}}$, where $S^2$ is the 
2-dimensional unit sphere in  $\mathbb{R}^3$, as our alphabet. 
Now the configuration space is the Cartesian product 
$\Omega =((S^2)^{\mathbb{Z}})^{\mathbb{N}},$ 
that is, a configuration is a point 
$\sigma=(\sigma(1),\sigma(2), \cdots)\in \Omega$,
where each ${\sigma}(i)$ is of the form 
$\sigma(i)=(\ldots,\sigma_{(i,-2)},\sigma_{(i,-1)},\sigma_{(i,0)}, \sigma_{(i,1)}, \ldots)$, 
and each $\sigma_{(i,j)}\in S^2$. We denote by $\|\cdot\|$ the $\mathbb{R}^3$
Euclidean norm, and $v\cdot w$ the inner product of two elements of $\mathbb{R}^3$.

Fix a summable ferromagnetic translation invariant 
 interaction  $J$ on $\mathbb{Z}$, that is,  
 a function $J:\mathbb{Z}\to (0,\infty)$ and assume that 
$J(n)= e^{-|n|\alpha}$, for some $\alpha>0$. Of course, we have  
$
\sum_{n\in \mathbb{Z}} J(n)<\infty.
$
Now we consider the potential $f:\Omega \to \mathbb{R}$ given by 
\begin{equation}\label{Heisenberg Potential}
f(\sigma)= \sum_{ n\in \mathbb{Z}} J(n)\ \sigma_{(1,n)}\cdot\sigma_{(2,n)}.
\end{equation}
Note that this potential has only first nearest neighbors interactions.

\begin{figure}[h]
\centering
\includegraphics[width=0.6\linewidth]{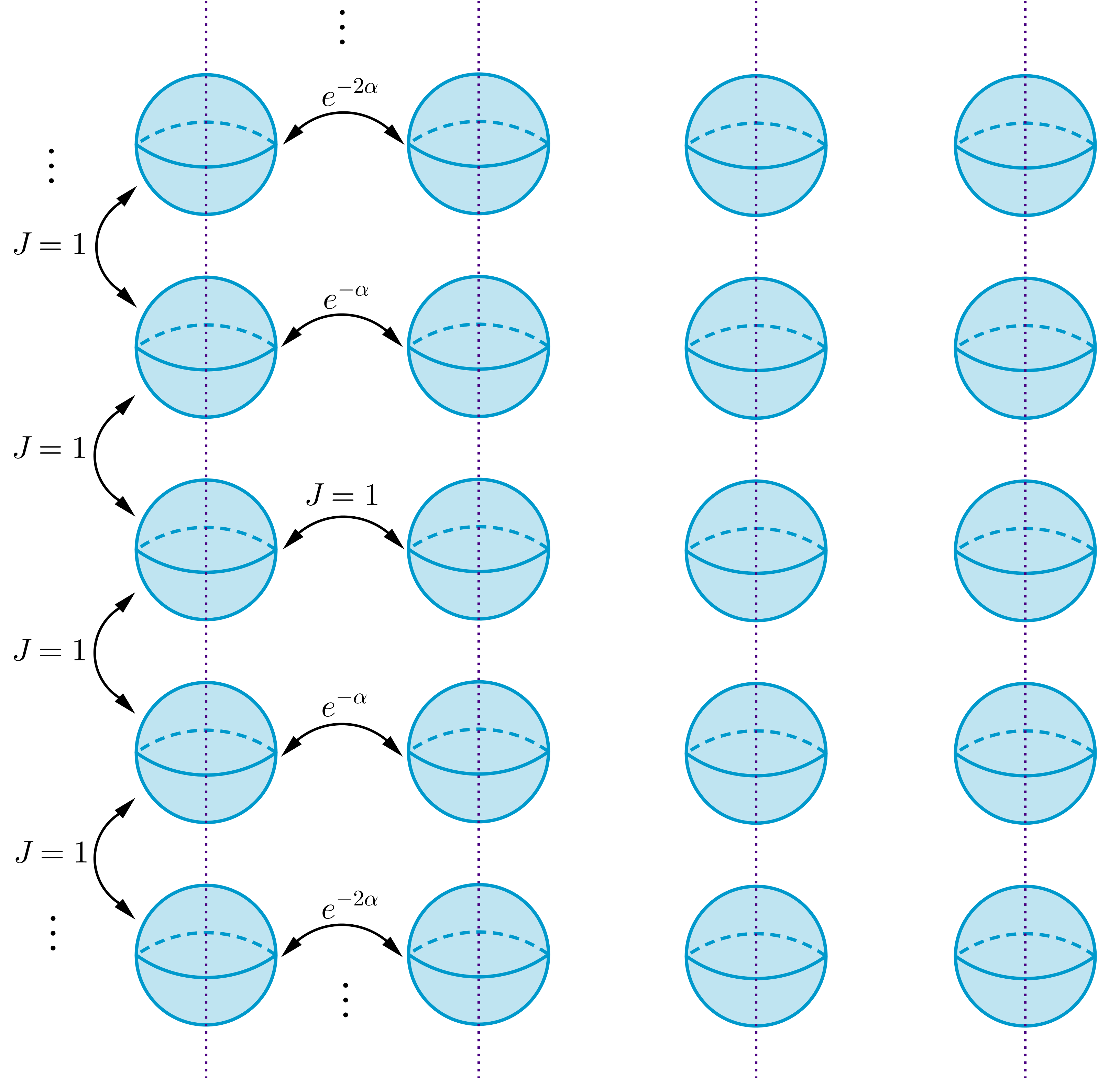}
\caption[Fig1]{The configuration space $((S^2)^{\mathbb{Z}})^{\mathbb{N}}$}
\label{fig:FigPaper}
\end{figure}

The  potential $f$ given by \eqref{Heisenberg Potential} is actually 
an $\alpha$-Hölder continuous function. 
Indeed, 
\begin{align}\label{Holder-1}
|f(\sigma)-f(\omega)|
&\leq 
\left |\sum_{n\in \mathbb{Z}}J(n)\sigma_{(1,n)}\cdot (\sigma_{(2,n)}-\omega_{(2,n)})\right|\nonumber
\\
&\hspace*{3.15cm} +
\left|\sum_{n\in \mathbb{Z}}J(n)\omega_{(2,n)}\cdot (\sigma_{(1,n)}-\omega_{(1,n)})\right|
\nonumber
\\[0.5cm]
&\leq
\sum_{n\in \mathbb{Z}}J(n)\|\sigma_{(2,n)}-\omega_{(2,n)}\|+ 
\sum_{n\in \mathbb{Z}}J(n)\|\sigma_{(1,n)}-\omega_{(1,n)}\| 
\end{align} 

From the very definition of the distance we have 

\begin{align}\label{Holder-2}
d(\sigma,\omega)
\geq 
\dfrac{1}{2^n}\sum_{j\in \mathbb{Z}}\dfrac{1}{2^{|j|}}\|\sigma_{(n,j)}-\omega_{(n,j)}\|
\geq 
\dfrac{1}{2^{n+|j|}}\|\sigma_{(n,j)}-\omega_{(n,j)}\|.
\end{align}
By using \eqref{Holder-1} and \eqref{Holder-2} we get that 
\[
\dfrac{|f(\sigma)-f(\omega)|}{d(\sigma,\omega)^{\alpha}}
\leq  
K_1\sum_{n\in \mathbb{N}}J(n)2^{\alpha n}
+
K_2\sum_{n\in \mathbb{N}}J(n)2^{\alpha n}.
\]
Since for all $n\in\mathbb{N}$ we have 
$
J(n)2^{\alpha n}
\leq 
\exp({-n(\alpha(1-\log 2))}),
$ 
and the constant $\alpha (1-\log 2)$ is  
positive, follows that the series $\sum_{n\in \mathbb{N}}J(n)2^{\alpha n}$ is  convergent.
Therefore, $f$ is an $\alpha$-H\"older continuous function.

From Theorem \ref{teo-RPF-compacto}, we have that there is a 
unique probability measure $\nu_{f}$ so that 
$\mathscr{L}_{f}^{*}\nu_f=\lambda_{f}\nu_f$. This probability 
measure, following \cite{CL16}, is a unique DLR-Gibbs measure associated 
to a quasilocal specification associated to $f$, see \cite{CL14} for the 
construction of this specification. 

By observing that the horizontal interactions, in our model, goes fast to zero in the $y$-direction 
is naturally to expect that the model is essentially a one-dimensional model.
This feature allow us to obtain the following result

\begin{corollary}
Let $\beta f$ be a potential, where $f$ is given by 
\eqref{Heisenberg Potential} and the inverse temperature $\beta \in (0,\infty)$.
Consider the a priori measure $\hat{\nu}$ on $(S^2)^{\mathbb{Z}}$, 
constructed from \eqref{Heisenberg-measure}.
Then for any fixed $\beta>0$ and $m\in\mathbb{Z}$ 
there are positive constants 
$K(\beta)$ and $c(\beta)$ such that for all $n\in\mathbb{N}$ we have
\[
\int_{\Omega} (\sigma_{(1,m)}\cdot\sigma_{(n+1,m)})\ d\nu_{\beta f}
\leq 
K(\beta)e^{-c(\beta)n}.
\]
Furthermore, the pressure functional is differentiable at $\beta f$ and 
its derivative is given by expression \eqref{pressure-derivative}.
\end{corollary}

\begin{proof}
Fix $m\in\mathbb{Z}$ and let 
$\sigma_{(n,m)}\equiv(\sigma_{(n,m)}^{x},\sigma_{(n,m)}^{y}, \sigma_{(n,m)}^{z})$.
Consider the following continuous potentials $\varphi^u, \xi^u$ given by
\[
\varphi^u(\sigma) = \sigma_{(1,m)}^{u}
\qquad\text{and}\qquad
\xi^u(\sigma) = \frac{\sigma_{(1,m)}^{u}}{h_{\beta f}(\sigma)}, \quad u=x,y,z.
\]
Note that 
\[
\int_{\Omega} \xi \, dm_{\beta f}
=
\int_{\Omega} \sigma_{(1,m)}^{u} \, d\nu_{\beta f}(\sigma) 
=
0,\quad u=x,y,z,
\]
where the last equality comes from the $O(3)$-invariance of 
the eingemeasure. 

Since $h_{\beta f}(\sigma)=h_{\beta f}(-\sigma)$, see \cite{CL17}, it follows again from the 
$O(3)$-invariance of $\nu_{f}$ that 
\[
\int_{\Omega} \sigma_{(1,m)}^{u} \, dm_{\beta f} 
= 
\int_{\Omega} \sigma_{(1,m)}^{u} h_f(\sigma) \, d\nu_{\beta f} 
=
0, \quad u=x,y,z.
\]
Therefore 
\begin{align}\label{Decay-Two-Point}
\int_{\Omega} (\sigma_{(1,m)}^{u}\sigma_{(n+1,m)}^u)\ d\nu_{\beta f}
&=
\int (\varphi^u\circ \sigma^n)\xi^u \, dm_{\beta f}
\nonumber
\\
&=
C_{\varphi^{u},\xi^{u},m_{\beta f}}(n)
=
O(e^{-c(\beta)n}),  
\end{align}
$u=x,y,z.$
By summing \eqref{Decay-Two-Point} with $u=x,y,z$ we get the claimed exponential decay.

The last statement follows from the fact that $(S^2)^{\mathbb{Z}}$ is a compact metric space, 
$\hat{\nu}$ is a full support a priori measure in
$(S^2)^{\mathbb{Z}}$ and $\beta f$ is an $\alpha$-Hölder.
Therefore the Proposition \ref{Decay Correlation} applies and 
corollary follows.
\end{proof}

\section*{Acknowledgments}
The author would like to thank  Leandro Cioletti, Artur Lopes
and Andr\'eia Avelar for fruitfull discussions and comments.

%\begin{flushleft}
%
%
%
%
%\bigskip
%
%
%{\sc 
%Eduardo Ant\^onio da Silva
%\\[0.2cm]
%Departamento de Matem\'atica\\
%Universidade de Bras\'ilia\\
%Campus Universit\'ario Darcy Ribeiro - Asa Norte\\
%70910-900  Bras\'ilia - DF - Brazil.}
%\\
%{\it eduardo23maf@gmail.com}
%\end{flushleft}

\end{document}